\documentclass[%
 reprint,
unsortedaddress,
nofootinbib,
 amsmath,amssymb,
 aps,
]{revtex4-2}
\usepackage{graphicx}
\usepackage{dcolumn}
\usepackage{bm}

\usepackage{enumerate}
\usepackage{dsfont}
\usepackage{amsfonts}
\usepackage{amssymb}
\usepackage[utf8]{inputenc}
\usepackage[T1]{fontenc}
\usepackage{graphicx}
\usepackage{mathtools}
\usepackage{bbm}
\usepackage{soul}
\usepackage[dvipsnames]{xcolor}
\usepackage[colorlinks=true]{hyperref}
\usepackage{amsthm}
\usepackage{amsmath}

\newcommand{\R}{\mathbb R}

\newcommand{\db}{\bar\partial}

\makeatletter
\newcommand*{\rom}[1]{\expandafter\@slowromancap\romannumeral #1@}
\makeatother

\newtheorem{prop}{Proposition}

\newcommand{\ab}{$(\alpha,\beta)$}

\begin{document}

\preprint{APS/123-QED}

\title{A Cosmological Unicorn Solution to Finsler Gravity}

\author{Sjors Heefer}
 \email{s.j.heefer@tue.nl}
 \affiliation{%
 Department of Mathematics and Computer Science, Eindhoven University of Technology, Eindhoven 5600MB, The Netherlands
}

\author{Christian Pfeifer}
\email{christian.pfeifer@zarm.uni-bremen.de}
\affiliation{ZARM, University of Bremen, 28359 Bremen, Germany}

\author{Antonio Reggio}
 \email{antonio.reggio1@gmail.com}
\author{Andrea Fuster}
 \email{a.fuster@tue.nl}
\affiliation{%
 Department of Mathematics and Computer Science, Eindhoven University of Technology, Eindhoven 5600MB, The Netherlands
}

\date{\today}

\begin{abstract}
We present a new family of exact vacuum solutions to Pfeifer and Wohlfarth's field equation in Finsler gravity, consisting of Finsler metrics that are Landsbergian but not Berwaldian, also known as unicorns due to their rarity. Interestingly we find that these solutions have a physically viable light cone structure, even though in some cases the signature is not Lorentzian but positive definite. We furthermore find a promising analogy between our solutions and classical Friedmann-Lema\^itre-Robertson-Walker cosmology. One of our solutions, in particular, has cosmological symmetry, i.e. it is spatially homogeneous and isotropic, and it is additionally conformally flat, with the conformal factor depending only on the timelike coordinate. We show that this conformal factor can be interpreted as the scale factor, we compute it as a function of cosmological time, and we show that it corresponds to a linearly expanding (or contracting) Finsler universe.
\end{abstract}

\maketitle

\section{Introduction}
The interest in Finsler-geometric modified theories of gravity has picked up in recent years, and rightly so. It has become clear that small, natural modifications in the classic axiomatic approach by Ehlers, Pirani and Schild (EPS) \cite{Ehlers2012} naturally lead to Finsler spacetime geometry; \cite{Bernal_2020,Lammerzahl:2018lhw,TAVAKOL198523}; it has also become clear that modified dispersion relations (MDRs), usually discussed in the context of quantum gravity phenomenology \cite{Addazi_2022} generically induce a Finsler geometry on spacetime \cite{Girelli:2006fw,Raetzel:2010je,Rodrigues:2022mfj,Lobo:2020qoa} and it has been conjectured that Finsler spacetime geometry describes the gravitational field of a kinetic gas more accurately compared to its usual treatment in the Einstein-Vlasov system, \cite{Hohmann:2019sni,Hohmann:2020yia}. These results, to name just a few, show clearly that in certain situations, for instance in Planckian regimes or for certain type of matter, it is to be expected that Finsler geometry should be the proper way to model spacetime. In addition to the applications in gravitational physics, see also \cite{CANTATA:2021ktz}, there are several other instances in physics, such as the description of the propagation of waves in media, where Finsler geometry seems to be the appropriate tool \cite{Pfeifer_2019}.\\

Various physical subclasses of Finsler spacetimes can be distinguished, and two particularly important ones are the class of Berwald spacetimes and the class of Landsberg spacetimes that may be thought of as incrementally non-(pseudo-)Riemannian, respectively (precise definitions will follow later). Every Berwald spacetime is also Landsberg, but whether or not the opposite is true has been a long-standing open question in Finsler geometry. In fact, Matsumoto has stated in 2003 that this question represents the next frontier of Finsler geometry \cite{Bao_unicorns}, and as a token of their elusivity, Bao \cite{Bao_unicorns}  has called these non-Berwaldian Landsberg spaces \textit{`[\dots] unicorns, by analogy with those mythical single-horned horse-like creatures for which no confirmed sighting is available.'} Since 2006 some examples of unicorns have been obtained by Asanov \cite{asanov_unicorns}, Shen \cite{shen_unicorns} and  Elgendi \cite{Elgendi2021a} by relaxing the definition of a Finsler space. Even such examples of so-called $y$-local unicorns are still exceedingly rare. \\

Here we present a new family of exact solutions to Pfeifer and Wohlfarth's Finslerian extension of Einstein's field equations \cite{Pfeifer:2011xi,Hohmann_2019} which is precisely such a unicorn. It falls into one of the classes introduced by Elgendi. Our solutions extend the very short list of known exact solutions in Finsler gravity. Indeed, to the best of our knowledge the only ones currently known in the literature are the (m-Kropina type) Finsler pp-waves \cite{Fuster:2015tua} and their generalization as Very General Relativity (VGR) spacetimes \cite{Fuster:2018djw}, the Randers pp-waves \cite{Heefer_2021}, and the pp-waves of general \ab-metric type \cite{Heefer2023b}.\\

Interestingly we find that these solutions have a physically viable light cone structure, even though in some cases the signature is not Lorentzian but positive definite. In fact, the Finslerian light cone turns out to be equivalent to that of the flat Minkowski metric. The set of unit timelike directions (technically the indicatrix inside the light cone), on the other hand, has either a hyperbolic or a spherical nature, depending of the signature within the timelike cone, as expected.

Furthermore we find a natural cosmological interpretation of one of our solutions and a promising analogy with classical Friedmann-Lema\^itre-Robertson-Walker (FLRW) cosmology. In particular, our solution has cosmological symmetry, i.e. it is spatially homogeneous and isotropic, and it is additionally conformally flat, with the conformal factor depending only on the timelike coordinate. We show that this conformal factor can be interpreted as the scale factor, we compute the scale factor as a function of cosmological time, and we show that it corresponds to a linearly expanding (or contracting) Finslerian universe.

\section{Finsler geometry}

Before we recall the basic notions of Finsler geometry below we first introduce some notation. Given a (spacetime) manifold $M$, which we assume to be four dimensional, and given some coordinates on $M$, we will always consider the natural induced coordinates on its tangent bundle $TM$. More precisely, given a coordinate chart $(U,x)$ on $M$, $U\subset M$, we obtain a coordinate chart $(TU,(x,y))$ on $TM$, $TU\subset TM$ of $TM$, where a point $Y = y^\mu\partial_\mu \in TM$ is labeled by $(x,y)$. By a slight abuse of notation we will generally identify any point in $M$ with its expression in coordinates, and similarly for points in $TM$, i.e. $Y=y$. We denote the coordinate basis vectors of the tangent spaces of $TM$ by $\partial_\mu = \partial/\partial x^\mu$ and $\bar \partial_\mu = \partial/\partial y^\mu$, where $\mu = 0,...,3$.\\

A Finsler space is a smooth manifold $M$ endowed with a Finsler metric, i.e. a smooth map $F:TM\setminus 0\to \R$ such that
\begin{itemize}
	\item $F$ is (positively) homogeneous of degree $1$ with respect to $y$:
	\begin{align}
	F(x,\lambda y) =\lambda F(x, y)\,,\quad \forall \lambda>0\,;
	\end{align}
	\item the \textit{fundamental tensor} 
        \begin{align}
            g_{\mu\nu} = \db_\mu\db_\nu \left(\tfrac{1}{2}F^2\right)
        \end{align} 
        is nondegenerate.
\end{itemize}
The fundamental tensor $g_{\mu\nu}$ depends generally on both $x$ and $y$. When $F^2$ is quadratic in $y$, or, equivalently, when $g_{\mu\nu}$ depends only on $x$, then $g_{\mu\nu}$ is a (pseudo-)Riemannian metric and the theory reduces to (pseudo-)Riemannian geometry. \\

In order to describe spacetime geometry, one usually demands that the signature of $g_{\mu\nu}$ be Lorentzian, at least in some conic open subset of $TM$, which one might hope to identify with the cone of timelike directions. Moreover, in applications one very often encounters Finsler structures\footnote{In the literature one finds various other, more stringent, definitions of Finsler spacetimes, going back to the original definition by Beem \cite{Beem}. They vary in their precise technical details, depending on the scope of the application, see e.g.  \cite{Javaloyes2018,Caponio_2020,Hohmann:2021zbt} and references therein.}  that are only properly defined (smooth, nondegenerate) on a subset of $TM\setminus 0$. Such Finsler metrics are sometimes referred to as $y$-local, as opposed to $y$-global \cite{Bao_unicorns}. In particular, the unicorn solution that we will present here is of $y$-local type.

\subsection{The nonlinear connection and geodesic spray}

The \textit{Cartan nonlinear connection} is the unique homogeneous (in general nonlinear) Ehresmann connection on $TM$ that is smooth on $TM\setminus\{0\}$, torsion-free and compatible with $F$. 
It may therefore be viewed as a generalization of the Levi-Civita connection. For details we refer e.g. to \cite{Szilasi}. Its connection coefficients are given by
\begin{align}
N^\mu_\nu = \frac{1}{4}\bar{\partial}_\nu \bigg(g^{\mu \rho}\big(y^\sigma \partial_\sigma\bar{\partial}_\rho F^2 - \partial_\rho F^2\big)\bigg)\, .
\end{align}
The nonlinear connection induces the horizontal derivatives
\begin{align}
    \delta_\mu = \partial_\mu -N^\nu_\mu\bar\partial_\nu,
\end{align}
that, together with the $\bar\partial_\mu$, span each tangent space $T_{(x,\dot x)}TM$.
The (geodesic) spray coefficients can then be defined as
\begin{align}
G^\mu \equiv N^\mu_\nu y^\nu = \frac{1}{2}g^{\mu\rho}\left(y^\sigma\partial_\sigma\bar\partial_\rho F^2 - \partial_\rho F^2\right),
\end{align}
where the second equality follows from Euler's theorem for homogeneous functions. It immediately follows that we also have $N^\mu_\nu = \tfrac{1}{2}\bar\partial_\nu G^\mu$. The importance of the spray coefficients comes from the fact that the geodesics of $F$ are given by
\begin{align}
    \ddot x^\mu + G^\mu(x,\dot x) = 0,
\end{align}
which coincidentally is also the autoparallel equation of the nonlinear connection $N^\mu_\nu$. \\


The curvature of the nonlinear connection is defined via $ R^\rho{}_{\mu\nu}\bar\partial_\rho = -[\delta_\mu,\delta_\nu]$, which implies that
\begin{align}
    R^\rho{}_{\mu\nu} =  \delta_\mu N^\rho_\nu-\delta_\nu N^\rho_\mu.  
\end{align}
From the nonlinear curvature one may define the Finsler Ricci scalar and Ricci tensor as follows
\begin{align}\label{eq:def_Ricci}
    \text{Ric} = R^\rho{}_{\rho\mu}y^\mu,\qquad R_{\mu\nu} = \tfrac{1}{2}\db_\mu \db_\nu\text{Ric}.
\end{align}
A Finsler space is said to be \textit{Ricci-flat} if Ric $=0$, or equivalently, $R_{\mu\nu}=0$. We remark that the Finsler Ricci scalar is not to be confused with the scalar curvature usually defined in (pseudo-)Riemannian geometry as $R = g^{\mu\nu}R_{\mu\nu}$, also sometimes called the Ricci scalar.

\subsection{The Chern-Rund connection}
In addition to the canonical nonlinear connection,  various canonical linear connections can be introduced. However, the price one has to pay for linearity is that the linear connections do not in general live on the vector bundle $TM$ but rather on its pull-back $\pi^*TM$ by the canonical projection $\pi:TM\to M$. The pull-back bundle $\pi^*TM$ is considered as a vector bundel over $TM\setminus 0$ and sections of this vector bundle may be thought of simply as vector fields on $M$ with a dependence on both $x$ and $y\neq 0$. Since the manifold $TM\setminus 0$ has dimension $2n$, we get in general two sets of linear connection coefficients, namely
\begin{align}
    \nabla_{\delta_\mu}\partial_\nu = \Gamma^\rho_{\nu\mu}\partial_\rho,\qquad \nabla_{\bar\partial_\mu}\partial_\nu = \bar\Gamma^\rho_{\nu\mu}\partial_\rho.
\end{align}
The Chern-Rund connection is the unique linear connection on $\pi^*TM$ that is torsion-free and \textit{almost} metric compatible. For details we refer e.g. to \cite{Bao}. These conditions imply that $\bar \Gamma^\rho_{\mu\nu}=0$ and
\begin{align}
    \Gamma^\rho_{\mu\nu} = \tfrac{1}{2}g^{\rho\sigma}\left(\delta_\mu g_{\sigma \nu} + \delta_\nu g_{\mu\sigma} - \delta_\sigma g_{\mu\nu}\right).
\end{align}
Notice the similarity to the formula for the Levi-Civita Christoffel symbols of a (pseudo-)Riemannian metric. From this it is immediately clear that the Chern-Rund connection reduces to the Levi-Civita connection when $F$ is (pseudo-)Riemannian. 

\subsection{Berwald and Landsberg spaces}

Next we introduce two important classes of Finsler spaces: Berwald spaces and Landsberg spaces. First, if the spray is quadratic in $y$, i.e.  $\bar\partial_\mu\bar\partial_\nu\bar\partial_\sigma G^\rho=0$ then $F$ is said to be of \textit{Berwald} type. What this means geometrically is that the Chern-Rund connection may be understood as an affine connection on $M$, i.e. equivalently, a space is Berwald if and only if the connection coefficients $\Gamma^\rho_{\mu\nu}$ of the Chern connection depend only on $x$. \\

And second, introducing the Landsberg curvature
\begin{align}\label{eq:Landsb_tensor}
S_{\mu\nu\sigma} =  -\tfrac{1}{4}y_\rho \, \bar\partial_\mu\bar\partial_\nu\bar\partial_\sigma G^\rho, 
\end{align}
and the mean Landsberg curvature $S_\sigma = g^{\mu\nu}S_{\sigma\mu\nu}$, we say that a space is (weakly) \textit{Landsberg} if the (mean) Landsberg curvature vanishes identically. The geometrical significance of the Landsberg tensor is somewhat more difficult to state in simple terms without introducing more machinery, so instead we refer e.g. to \cite{Bao_unicorns}.\\

It is immediately obvious from the definitions that any Berwald space is a Landsberg space. Also, a (pseudo-)Riemannian space is always Berwald, hence in particular any (pseudo-)Riemannian space is Landsberg.

\subsection{Unicorns in Finsler geometry}

As observed above, we have the following inclusions:
\begin{align*}
   &\text{(pseudo-)Riemannian} \subset \text{Berwald} \subset \text{Landsberg}. 
\end{align*}
It has been a long-standing open question whether the last inclusion is strict. Do there exist Landsberg spaces that are not Berwald? In the $y$-global case the answer is unknown. For $y$-local spaces some examples are known, but these are exceedingly rare. As such, non-Berwaldian Landsberg spaces are referred to as \textit{unicorns} \cite{Bao_unicorns}. We recommend \cite{Bao_unicorns,UnicornSurvey} for reviews on the unicorn problem. \\

The first unicorns were found by Asanov \cite{asanov_unicorns} in 2006 and his results were generalized by Shen \cite{shen_unicorns} a few years later. These were the only known examples of unicorns until Elgendi very recently provided some additional examples of unicorns \cite{Elgendi2021a}. One of the families of unicorns introduced by Elgendi will be central in this work.


\section{The Finslerian field equations}

Although various proposals for Finslerian field equations in vacuum can be found in the literature \cite{Rutz,Pfeifer:2011xi,Horvath1950,Horvath1952, Ikeda1981, Asanov1983, Chang:2009pa,Kouretsis:2008ha,Stavrinos2014,Voicu:2009wi,Minguzzi:2014fxa}, it seems fair to say that Pfeifer and Wohlfarth's field equation \cite{Pfeifer:2011xi} has the most robust foundation.  It is obtained as the Euler-Lagrange equation of the natural Finsler generalization of the Einstein Hilbert action \cite{Pfeifer:2011xi,Hohmann_2019}, and furthermore it has been shown recently that the equation is the \textit{variational completion} of Rutz's equation  \cite{Rutz}, Ric $= 0$. The latter is arguably the simplest and cleanest proposal, and well physically motivated, but it cannot be obtained by extremizing an action functional, complicating the coupling of the theory to matter. For reference, Einstein's vacuum equation in the form $R_{\mu\nu} - \frac{1}{2}g_{\mu\nu}R = 0$ is also precisely the variational completion of the simpler equation $R_{\mu\nu}=0$ \cite{Voicu_2015}. While in the GR case the completed equation happens to be equivalent to the former, this is not true any longer in the Finsler setting.\\


Pfeifer and Wohlfarth's field equation in vacuum reads
\begin{align}
\text{Ric} \,- \,&\frac{F^2}{3}g^{\mu\nu}R_{\mu\nu} \nonumber \\
%
-\,&\frac{F^2}{3}g^{\mu\nu}\left(\bar\partial_\mu\dot S_\nu - S_\mu S_\nu + \nabla_{\delta_\mu}S_\nu \right) = 0, \label{eq:Pfeifer_Wohlfarth_eq}
\end{align}
where $\dot S_\nu \equiv y^\rho\nabla_{\delta_\rho} S_\nu$. 
%
%
For (pseudo-)Riemannian metrics, \eqref{eq:Pfeifer_Wohlfarth_eq} reduces to Einstein's field equation in vacuum. From the general expression \eqref{eq:Pfeifer_Wohlfarth_eq} it becomes immediately apparent that for weakly Landsberg spaces, characterized by the defining property that $S_i=0$, the field equation in vacuum attains the relatively simple form
\begin{align}\label{eq:Pfeifer_Wohlfarth_eq_Landsberg}
\text{Ric} - \frac{F^2}{3}g^{\mu\nu}R_{\mu\nu} = 0.
\end{align}
Recalling the definition \eqref{eq:def_Ricci} of $R_{\mu\nu}$ we have the following immediate result:

\begin{prop}
    Any Ricci-flat, $Ric=0$, weakly Landsberg space is a solution to the field equations \eqref{eq:Pfeifer_Wohlfarth_eq}
\end{prop}

In other words, any weakly Landsberg solution to the Rutz equation is automatically a solution to \eqref{eq:Pfeifer_Wohlfarth_eq}.


\section{An exact unicorn solution to Finsler gravity} 

\subsection{Elgendi's class of unicorns}

Elgendi recently introduced a class of unicorns \cite{Elgendi2021a} with Finsler metric given by
\begin{equation}
    F = \left(a\beta+\sqrt{\alpha^2-\beta^2}\right)e^{\frac{a\beta}{a\beta+\sqrt{\alpha^2-\beta^2}}},
\end{equation}
in terms of a real, nonvanishing constant $a$ and 
\begin{align}
    \alpha = f(x^0)\sqrt{(y^0)^2+\phi(\hat y)}, \qquad \beta = f(x^0) y^0,
\end{align}
where $f$ is a positive real-valued function and $\phi(\hat y) = \phi_{ij} \hat y^i \hat y^j = \phi_{ij} y^i  y^j$ is a nondegenerate quadratic form on the space spanned by $\hat y = (y^1, y^2, y^3)$, with constant, symmetric coefficients $\phi_{ij}$. Here and in what follows, indices $i,j,\dots$ will run over $1,2,3$, whereas greek induces $\mu,\nu,\dots$ will run over $0,1,2,3$. From a gravitational physics perspective, the only degree of freedom of these Finsler metrics is the function $f(x^0)$. The geodesic spray of $F$ is given by
\begin{align}
    G^0 &= \left(\frac{2f(x^0)^2(y^0)^2-\alpha^2}{f(x^0)^2}+\frac{a^2-1}{a^2}\frac{\alpha^2-\beta^2}{f(x^0)^2}\right)\frac{f'(x^0)}{f(x^0)}\\
    G^{i} &= Py^i,
\end{align}
where 
\begin{align}
    P &= 2\left(y^0+\frac{1}{af(x^0)}\sqrt{\alpha^2-\beta^2}\right)\frac{f'(x^0)}{f(x^0)},
\end{align}
and the Landsberg tensor vanishes identically. This shows that these metrics are indeed Landsberg, but not Berwald, since the $i$-components of the spray are not
quadratic in $y$. Note that our $G^k$ is twice the $G^k$ in Elgendi's paper \cite{Elgendi2021a}, due to a difference in convention. Explicitly then, Elgendi's unicorns have the form
\begin{align}\label{eq:unipos}
     F = f(x^0)\left(y^0 + \sqrt{\phi(\hat y)}\right)e^{\frac{y^0}{y^0 + \sqrt{\phi(\hat y)}}}.
\end{align}
where we have absorbed the constant $a$ into a redefinition of $x^0$. For our purposes we will modify this expression slightly, though.
\subsection{The modified unicorn metric}

The expression \eqref{eq:unipos} defining the unicorn metric is only well defined whenever $\phi(\hat y)\geq 0$. If $\phi$ is positive definite, this is necessarily the case, but in other signatures this is not always true. In order to extend the domain of definition of $F$, an obvious first approach would naturally be to replace $\phi$ by its absolute value, $|\phi|$, i.e.

\begin{align}\label{eq:Uni}
    F = f(x^0) \left(y^0 + \sqrt{|\phi(\hat y)|}\right) e^{\frac{y^0}{y^0 + \sqrt{|\phi(\hat y)|}}}\,.
\end{align}


From the physical point of view this is still not completely satisfactory, however. This can be seen by considering the light cone corresponding to such a Finsler metric $F$, given by the set of vectors for which $F=0$, and interpreted as the set of propagation directions of light rays. Indeed, barring some potential problems with the exponent to which we will come back later, the light cone would be given by those vectors satisfying $y^0 = -\sqrt{|\phi|}$, which would imply that the light cone is contained entirely within the half space $y^0<0$. Interpreting $y^0$ for the moment as a time direction, this would have the result that light rays can only propagate `backward in time' (with regards to their parameterization). Alternatively one might interpret this as saying that light can only be received, not emitted. This would mean that we could not describe radar signals or use radar methods \cite{Perlick2008,Pfeifer:2014yua,Gurlebeck:2018nme,Heefer2023b}, which we certainly can in the physical world. If on the other hand $y^0$ is a spacelike coordinate, the situation becomes even worse, as this would mean that light cannot propagate in the spatial $y^0$-direction (in contrast to the $-y^0$-direction).\\

To obtain a viable light cone structure, which allows for emission and reception of light rays in all spatial directions at each $x\in M$, we consider a modified unicorn metric, inspired by the construction of modified Randers metrics in \cite{Heefer2023b}. \\

Thus, our starting point will be the following Finsler metric:
\begin{align}\label{eq:uni_spacetime_pre}
\boxed{
    F_{0}
    = f(x^0) \left( |y^0| + \text{sgn}(\phi)\sqrt{|\phi|} \right) e^{\frac{|y^0|}{ |y^0| + \text{sgn}(\phi)\sqrt{|\phi|} }}\,,
    }
\end{align} 
from which we define the \textit{modified unicorn metric} as:
\begin{align}\label{eq:uni_spacetime}
    \boxed{
    F = \left\{\begin{matrix}
    F_{0} & \text{if} & |y^0| + \text{sgn}(\phi)\sqrt{|\phi|} \neq 0\\
    0 & \text{if} &  |y^0| + \text{sgn}(\phi)\sqrt{|\phi|} = 0
    \end{matrix}\right. \,.}
\end{align}
It will be confirmed in Sec.\,\ref{sec:solvFgrav} that such a metric is still of unicorn type, i.e. Landsberg but not Berwald.\\

The case distinction is necessary since the metric \eqref{eq:uni_spacetime_pre} is ill defined at vectors satisfying $|y^0| + \text{sgn}(\phi)\sqrt{|\phi|}=0$, because of the division by the same number in the exponent. $F_0$ does not have a well-defined limit to such vectors either, because the exponent does not stay negative in such a limit. This issue was already present for Elgendi's unicorn \eqref{eq:unipos}, yet from a purely mathematical point of view, this is not necessarily a problem, as one can simply opt to exclude this set of vectors from the domain of $F$. From a physics perspective, however, we want to interpret the set $F=0$ as directions in which light propagates.

Our definition of $F$ as in \eqref{eq:uni_spacetime}, ensures the existence of a well-defined light cone and preserves the unicorn property, which is why we will work with this throughout this article, even though it leads to a discontinuity at the light cone, which we discuss in a little more detail in Sec.\,\ref{sec:detg}.

Next, we will determine the cone structure and signature of the fundamental tensor, which will turn out to depend on the signature of $\phi$. In Sec.\,\ref{sec:solvFgrav} 
we determine the free function $f(x^0)$ in \eqref{eq:uni_spacetime} by solving the Finsler gravity equation. Afterward, we discuss the interpretation of the deformed unicorn in the context of cosmology and we conclude.


\subsubsection{Cone structure}

First we observe that, regardless of the exact form or signature of $\phi$, our modified unicorn metrics have a light cone structure $F=0$ that is equivalent to that of a pseudo-Riemannian metric. 

\begin{prop}
    The light cone of the modified unicorn metric \eqref{eq:uni_spacetime} is given by 
    \begin{align}\label{eq:lightcone}
        \left(y^0\right)^2 + \phi = 0.
    \end{align}
\end{prop}
\begin{proof}
    The result follows from the following sequence of equivalences.
    \begin{align}
        F = 0 &\Leftrightarrow \text{sgn}(\phi)\sqrt{|\phi|}+|y^0|=0 \\
        &\Leftrightarrow  \text{sgn}(\phi)\sqrt{|\phi|}=-|y^0|\\
        &\Leftrightarrow  |\phi|=\left(y^0\right)^2 \quad\text{and}\quad \phi \leq 0\\
        &\Leftrightarrow  \phi=-\left(y^0\right)^2\\
        &\Leftrightarrow  \phi+\left(y^0\right)^2 = 0.
    \end{align}
\end{proof}
Depending on the signature of $\phi$, we can make a more precise statement. 

\begin{prop}\label{prop:lightcone_details}
    Let $F$ be the modified unicorn metric \eqref{eq:uni_spacetime} corresponding to some some nondegenerate quadratic form $\phi_{ij}$. Then the following holds:
    \begin{itemize}
        \item For $\phi_{ij}$ of signature $(+,+,-)$ the null structure of $F$ is identical to the Minkowski metric spacetime light cone structure of signature $(+,+,+,-)$.
        \item If $\phi_{ij}$ is negative definite, i.e. of signature $(-,-,-)$ then the light cone of $F$ is identical to the Minkowski metric spacetime light cone of signature $(+,-,-,-)$.
        \item If $\phi_{ij}$ is positive definite, i.e. of signature $(+,+,+)$ then the light cone of $F$ is given by $y^\mu=0$.
        \item For $\phi_{ij}$ of signature $(-,-,x)$ the null structure of $F$ is identical to the one of a pseudo-Riemannian metric manifold with signature $(+,-,-,+)$.
    \end{itemize}
\end{prop}

This singles out the $(+,+,-)$ and $(-,-,-)$ signatures of $\phi_{ij}$ as the ones that are physically reasonable. Using the light cone structure it is natural to interpret the interior of the future and past pointing light cone as cones of future and past pointing timelike directions, respectively. In the $(+,+,-)$ case this leads to the interpretation of the coordinate $x^3$ as timelike coordinate, while in the $(-,-,-)$ case $x^0$ would be the timelike coordinate.


\subsubsection{Signature of the fundamental tensor}\label{sec:detg}

Next we investigate the signature of our modified unicorn metrics.

\begin{prop}\label{prop:signature}
Consider a modified unicorn metric $F$ as in \eqref{eq:uni_spacetime} and let $\mathcal{S(\phi)}$ be the set of all $\hat y = (y^1,y^2,y^3)$ that are $\phi$-spacelike and $\mathcal{T(\phi)}$ the set of all $\hat y$ that are $\phi$-timelike.
\begin{itemize}
    \item If $\phi$ is positive definite or negative definite then $g_{\mu\nu}(x,y)$ is positive definite on its entire domain of definition.
    \item If $\phi$ is Lorentzian then $g_{\mu\nu}(x,y)$ is of Lorentzian signature $(+,+,+,-)$ for all $y\in \R\times\mathcal S(\phi)$ and $g_{\mu\nu}(x,y)$ is of signature $(+,-,-,+)$ for all $y\in  \R\times \mathcal T(\phi)$.
\end{itemize}
\end{prop}

The proof of this proposition can be found in the appendix \ref{app:prf4}.  Let us discuss here some interesting observations.

Consider the two physically reasonable scenarios we identified below Prop. \ref{prop:lightcone_details} as a result of their good light cone structure, i.e. $\phi$ having $(+,+,-)$ or $(-,-,-)$ signature. In both cases the light cone is equivalent to that of Minkowski space. 

Surprisingly, inside the interior of this cone it is easy to see from the previous proposition, that the signature of $g$ is not Lorentzian. Indeed, if $\phi$ has signature $(+,+,-)$, then $g$ has signature $(+,-,-,+)$ inside the cone, while if $\phi$ is negative definite then $g$ is positive definite inside the cone. In the latter case this does not contradict the existence of the light cone structure due to the reduced smoothness and continuity of the Finsler metric $F$ \eqref{eq:uni_spacetime}. Indeed, 
\begin{align}
    \lim_{\substack{|y^0| + \text{sgn}(\phi)\sqrt{|\phi|}\to 0\\ |y^0| + \text{sgn}(\phi)\sqrt{|\phi|}\leq 0}}F_0 = 0.
\end{align}
whereas
\begin{align}
    \lim_{\substack{|y^0| + \text{sgn}(\phi)\sqrt{|\phi|}\to 0\\ |y^0| + \text{sgn}(\phi)\sqrt{|\phi|} > 0}}F_0 = \infty\,.
\end{align}
In other words, for $\phi$ being negatively definite, $F$ extends continuously to the light cone from the spacelike directions, but not from the timelike directions. For $\phi$ being of signature $(+,+,-)$, the situation is reversed.

In several alternative definitions of Finsler spacetimes \cite{Javaloyes2018,Hohmann:2021zbt}, one requires that there exists a cone inside which the fundamental tensor has Lorentzian signature, in order to guarantee (among other things) the existence of a physical light cone structure. Here, however, we found that there exist Finsler geometries that do have a satisfactory light cone structure, even without this property. This is an interesting new observation about Finsler geometry in its own right: apparently even in positive definite signature, a light cone structure  may arise due to irregularities (e.g. discontinuities) of the Finsler metric. \\

We note, however, that the signature anomaly of the Finsler metric in the physically relevant $(-,-,-)$ case leads to the fact that the future and past pointing unit-normalized timelike directions each do not lie on a deformed  hyperboloid, but on a deformed sphere with one point removed, the zero vector. Moreover, as a result of the discontinuity in the Finsler metric, the norm of the timelike directions do not tend to zero when one approaches the light cone from the interior of the cone. Whether this is acceptable or poses a fundamental problem in the application of these metrics in spacetime physics is not immediately clear. At the very least,  the clock postulate (the fact that the time an observer measures between two events is given by the length of its worldline connecting these events, independent of the parametrization of the worldline) is still perfectly valid, due to the 1-homogeneity of the Finsler metric.

\subsection{Solving the Finsler gravity equation}\label{sec:solvFgrav}
Next, we seek to determine the form of the free function $f(x^0)$ of \eqref{eq:uni_spacetime} from the Finsler gravity equation \eqref{eq:Pfeifer_Wohlfarth_eq_Landsberg}. The geodesic spray of $F$ is given explicitly by
\begin{align}
    G^0 &= \left((y^0)^2 - |\phi|\right)\frac{f'(x^0)}{f(x^0)}\\
    G^{i} &= Py^i,\qquad i=1,2,3
\end{align}
where 
\begin{align}
    P = 2\left(|y^0|+\text{sgn}(\phi)\sqrt{|\phi|}\right)\text{sgn}(y^0)\frac{f'(x^0)}{f(x^0)}\,,
\end{align}
As for the original unicorn metrics \eqref{eq:unipos}, this geodesic spray is not quadratic in $y$, so our modified unicorn metrics \eqref{eq:uni_spacetime} are not Berwald. Before actually solving the field equations we will also confirm that the Landsberg tensor still vanishes, even after our modification, justifying the name \textit{unicorn}. To this end we employ the definition \eqref{eq:Landsb_tensor} of the Landsberg tensor and note that only the $G^i$ terms ($i=1,2,3$) can give a nontrivial contribution. For these terms we compute that
\begin{align}
    \bar\partial_\mu\bar\partial_\nu\bar\partial_\sigma G^i = \bar\partial_\mu\bar\partial_\nu\bar\partial_\sigma\sqrt{|\phi|}y^i + 3\delta^ i{}_{(\mu}\bar\partial_\nu\bar\partial_{\sigma)}\sqrt{|\phi|} .
\end{align}
To find the Landsberg tensor we need to contract this with $y_i = g_{i\mu}y^\mu = \tfrac{1}{2}\bar\partial_i F^2$, and it can be checked in a straightforward way that the latter can be written as some function times $\bar\partial_i\sqrt{|\phi|}$. It thus suffices to show that $\bar\partial_i\sqrt{|\phi|}\bar\partial_\mu\bar\partial_\nu\bar\partial_\sigma G^i=0$. Indeed, whenever $\mu\neq 0$, it follows from the homogeneity of $\sqrt{|\phi|}$ that the latter is equal to (for $\mu=0$ it vanishes immediately since $\phi$ does not depend on $y^0$)
\begin{align}
    &\bar\partial_i\sqrt{|\phi|}\left(\bar\partial_\mu\bar\partial_\nu\bar\partial_\sigma\sqrt{|\phi|}y^i + 3\delta^ i{}_{(\mu}\bar\partial_\nu\bar\partial_{\sigma)}\sqrt{|\phi|} \right) \\
    &= \sqrt{|\phi|}\bar\partial_\mu\bar\partial_\nu\bar\partial_\sigma\sqrt{|\phi|} + 3\bar\partial_{(\mu}\sqrt{|\phi|}\bar\partial_\nu\bar\partial_{\sigma)}\sqrt{|\phi|} \\
    &= \tfrac{1}{2}\bar\partial_\mu\bar\partial_\nu\bar\partial_\sigma \left(\sqrt{|\phi|}\right)^2 = \tfrac{1}{2}\bar\partial_\mu\bar\partial_\nu\bar\partial_\sigma |\phi| = 0,
\end{align}
which vanishes because $\phi$ is quadratic. Hence our modified unicorns are indeed non-Berwaldian Landsberg metrics, justifying their name.

\begin{prop}\label{prop:ricci_flatness}
    $F$ is Ricci-flat if and only if $f$ has the form $f(x^0) = c_1 \exp\left(c_2 x^0\right)$.
\end{prop}
\begin{proof}
    By definition, and using homogeneity and the fact that $N^\mu_\nu = \tfrac{1}{2}\bar\partial_\nu G^\mu$, we have
    \begin{align}
    %
    \text{Ric} &= R^\mu{}_{\mu\nu}y^\nu = (\delta_\mu N^\mu_\nu-\delta_\nu N^\mu_\mu) y^\nu \\
    &=  y^\nu((\partial_\mu - N^\rho_\mu\bar\partial_\rho )N^\mu_\nu-(\partial_\nu -N^\rho_\nu\bar\partial_\rho)N^\mu_\mu)\\
    &=  y^\nu \partial_\mu N^\mu_\nu-  y^\nu\partial_\nu N^\mu_\mu  \\
    &\qquad  -  y^\nu N^\rho_\mu\bar\partial_\rho N^\mu_\nu  + y^\nu N^\rho_\nu\bar\partial_\rho N^\mu_\mu  \\
    &=  \tfrac{1}{2}\left(y^\nu\partial_\mu \bar\partial_\nu G^\mu-  y^\nu\partial_\nu \bar\partial_\mu G^\mu\right) \\
    &\qquad  -  \tfrac{1}{4}\left(y^\nu\bar\partial_\mu G^\rho\bar\partial_\rho\bar\partial_\nu G^\mu  - y^\nu \bar\partial_\nu G^\rho\bar\partial_\rho\bar\partial_\mu G^\mu\right)\\
    &=  \partial_\mu  G^\mu-  \tfrac{1}{2}y^\nu\partial_\nu \bar\partial_\mu G^\mu \\
    &\qquad  -  \tfrac{1}{4}\left(\bar\partial_\mu G^\rho\bar\partial_\rho G^\mu  - 2 G^\rho\bar\partial_\rho\bar\partial_\mu G^\mu\right).
    \end{align}
    Using the identities
    \begin{align}
        \begin{array}{ll}
            \bar\partial_0 P = 2 f'/f, &\quad \bar\partial_0 G^0 = 2y^0f'/f,   \\        
            \bar\partial_0^2  G^0 = 2 f'/f, &\quad \bar\partial_0\bar\partial_i G^0 = \bar\partial_0\bar\partial_i P = \bar\partial_0^2 P = 0,   \\
            y^i\bar\partial_i G^0= -2|\phi| f'/f, &\quad  y^i\bar\partial_i P= 2 \text{sgn}(\phi y^0)\sqrt{|\phi|} f'/f,
        \end{array}
    \end{align}
    one finds after some slightly tedious manipulations that the last two terms in the expression for the Ricci tensor can both be expressed as
    \begin{align}
        \bar\partial_\mu G^\rho\bar\partial_\rho G^\mu  = 2 G^\rho\bar\partial_\rho\bar\partial_\mu G^\mu = n P^2,
    \end{align}
    where $n = \dim M = 4$ in our case. Hence these terms cancel each other out precisely. Denoting $G^0 = \bar G^0 f'/f$ and $P = \bar P f'/f$, so that $\bar G^0$ and $\bar P$ do not depend on $x^\mu$, and using the fact that $\bar\partial_\mu G^\mu = n P$, one finds furthermore that 
    \begin{align}
    \partial_\mu G^\mu &= (\partial_0^2\log |f|)\bar G^0, \\
     y^\nu\partial_\nu \bar\partial_\mu G^\mu &=n y^0(\partial_0^2\log |f|)\bar P.
    \end{align}
    Consequently we have
    \begin{align}
        \text{Ric} &=  \partial_\mu  G^\mu-  \tfrac{1}{2}y^\nu\partial_\nu \bar\partial_\mu G^\mu \\
        &= +(\partial_0^2\log |f|)\\
        &\,\,\times \left((1-n)(y^0)^2 -n |y^0| \text{sgn}(\phi)\sqrt{|\phi|} - |\phi| \right), \label{eq:Ricci_check_sgn}
    \end{align}
    which in dimension $n=4$ reduces to
    \begin{align}
        \text{Ric} &= -(\partial_0^2\log |f|)\\
        &\,\,\times \left(3(y^0)^2 +4 |y^0| \text{sgn}(\phi)\sqrt{|\phi|} + |\phi| \right). \label{eq:Ricci_check_sgn_2}
    \end{align}
    If Ric $=0$ then we must in particular have 
    \begin{align}
        0=\bar\partial_0^2\text{Ric} = -6(\partial_0^2\log |f|).
    \end{align}
    It thus follows that Ric $=0$ if and only if $\partial_0^2\log |f| =0$, the general solution to which is given by the stated form of $f$.
\end{proof}

This shows that the following family of Finsler metrics are exact vacuum solutions to Pfeifer and Wohlfarth's field equations in Finsler gravity:
\begin{align}\label{eq:unicorn_solution}
     \boxed{F \!= \!c_1e^{c_2 x^0}\!\left( |y^0|\! +\! \text{sgn}(\phi)\sqrt{|\phi|} \right)\! \exp\!\left(\!\frac{|y^0|}{ |y^0| \!+ \!\text{sgn}(\phi)\sqrt{|\phi|} }\!\right)\!,}
\end{align}
where $\phi=\phi(\hat y) = \phi_{ij}y^i y^j$, with $\phi_{ij}$ being a three-dimensional nondegenerate, symmetric bilinear form with constant coefficients and signature $(+,+,-)$ or $(-,-,-)$.\\

In fact it turns out that \textit{any} solution of the type \eqref{eq:uni_spacetime} must have this form.
\begin{prop}\label{prop:solution}
    A Finsler metric of the form \eqref{eq:uni_spacetime} is a solution to the Finslerian field equations in vacuum, Eq. \eqref{eq:Pfeifer_Wohlfarth_eq_Landsberg}, if and only if it can be written  as \eqref{eq:unicorn_solution}.
\end{prop}

Before we give a sketch of the proof, it is important to point out the exact meaning of the word `locally' in the proposition, as it will have essential consequences for the physical viability of such solutions. Of course, the word applies first and foremost to the $x$-coordinates in the usual sense, but it also applies to the tangent space coordinates $y^i$: there is \textit{a priori} no reason why one could not pick, at some point $x^\mu$, say, a certain $\phi_{ij}$ in a certain subset of $T_xM$ and a different $\phi_{ij}$ in a different subset of $T_xM$. Of course this could in general result in not having smoothness across the interface of the two regions, but this does not necessarily have to be a problem, unless one sets very strict smoothness requirements.\\
%
%
For consistency the different parts of $TM$ must satisfy some conditions. It is most natural to require that any such part be a conic subbundle, i.e. an open conic subset with nonempty fiber at each point in $M$.

\begin{proof}
    In four spacetime dimensions the proof is straightforward and most easily performed in convenient coordinates in which $\phi$ is diagonal with all nonvanishing entries equal to $+1$ or $-1$. From \eqref{eq:uni_spacetime} one can directly compute $g_{\mu\nu}$ and then its inverse $g^{\mu\nu}$. From \eqref{eq:Ricci_check_sgn_2} together with \eqref{eq:def_Ricci} one can immediately compute $R_{\mu\nu}$. We omit the intermediate expressions because they are somewhat lengthy, but plugging all of this into the field equation \eqref{eq:Pfeifer_Wohlfarth_eq_Landsberg} leads to
    \begin{align}
        \frac{-\partial_0^2\log |f|}{3\sqrt{|\phi|}}\left(-4\,\text{sgn}(\phi)|y^0|^3 -5y^0 \sqrt{|\phi|} + |\phi|^{3/2} \right) = 0.
    \end{align}
    Clearly this equation can only be satisfied for all $y^\mu$ in an open set where the fundamental tensor has Lorentzian signature if $\partial_0^2\log |f|=0$, in which case (the last sentence in the proof of) Prop. \ref{prop:ricci_flatness} shows that $F$ is in fact Ricci-flat and therefore must have the form \eqref{eq:unicorn_solution}.
    
\end{proof}

With this we found an exact solution to the Finsler gravity equations \eqref{eq:unicorn_solution}, starting from a generalized version of Elgendi's unicorns \eqref{eq:uni_spacetime}.

\subsection{Physical interpretation: A linearly expanding universe}
Having analyzed the mathematical properties of the unicorn Finsler spacetimes \eqref{eq:uni_spacetime}, and found the exact unicorn vacuum solution \eqref{eq:unicorn_solution} of the Finsler gravity vacuum equations \eqref{eq:Pfeifer_Wohlfarth_eq_Landsberg}, we now turn to the physical interpretation of this solution. We find that these Finsler gravity vacuum solutions yield a vacuum cosmology with linear time dependence of the scale factor.\\

To arrive at this conclusion we highlight the following properties of the the unicorn Finsler spacetimes \eqref{eq:uni_spacetime}:
\begin{itemize}
    \item Conformal flatness, with a conformal factor that is only spacetime dependent 
    \begin{align}
        F(x,y) = f(x) F_0(y)\,.
    \end{align}
    \item Cosmological symmetry, for the case when $\phi_{ij}$ has signature $(-,-,-)$, since then, by introducing spatial spherical coordinates $(r,\theta,\phi)$, we can write
    \begin{align}
        F(x,y) &= F(x^0,y^0,w), \\
        w^2 &= (y^1)^2+ (y^2)^2+(y^3)^2 \\
        &= (y^r)^2 + r^2 \left( (y^\theta)^2 + \sin^2\theta (y^\phi)^2\right)\,,
    \end{align}
    which is precisely of the form of a spatial flat homogeneous and isotropic Finsler geometry \cite{Hohmann:2020mgs}. This construction does not work for $\phi_{ij}$ with signature $(+,+,-)$, since then $y^3$, and not $y^1$ would be the timelike direction.
\end{itemize}
Combining these observations, we find that the unicorn Finsler spacetimes \eqref{eq:uni_spacetime} are of the form
\begin{align}\label{eq:conformalcomso}
    F(x^0, y^0,w) = f(x^0) F_0(y^0,w)\,.
\end{align}
This form reminds one immediately of classical flat FLRW spacetimes in conformal time $\eta$, which we identify here with the $x^0$ coordinate. When written in the language of Finsler geometry they are of the form \eqref{eq:conformalcomso} with
\begin{align}
    F_{0FLRW} =\sqrt{ (- (y^0)^2 + (y^1)^2 (y^2)^2 + (y^3)^2)}\,.
\end{align}
A redefinition of the time coordinate via $\frac{\partial \tilde x^0}{\partial x^0} = f(x^0)$, which implies that $\tilde y^0 = f(x^0) y^0$, then leads to the standard form of flat FLRW geometry
\begin{align}
    F = \sqrt{ (- (\tilde y^0)^2 +f(\tilde x^0)^2 \left( (y^1)^2 + (y^2)^2 + (y^3)^2)\right)}\,,
\end{align}
where the conformal factor is nothing but the usual cosmological scale factor and $\tilde x^0$ is the usual cosmological time $t$.

For the Finsler metric \eqref{eq:uni_spacetime} we employ the coordinate change $\frac{\partial \tilde x^0}{\partial x^0} = f(x^0)$, implying that $\tilde y^0 = f(x^0) y^0$, so that 
\begin{align}
     F  
    =  \left( |\tilde y^0| + \text{sgn}(\phi) f(\tilde x^0) \sqrt{|\phi|} \right) e^{\frac{|\tilde y^0|}{ |\tilde y^0|  + \text{sgn}(\phi) f(\tilde x^0) \sqrt{|\phi|} }}\,.
\end{align}
Hence, as in the classical FLRW geometry case, the conformal factor can be interpreted as scale factor of the spatial universe, $x^0$ as conformal time $\eta$ and $\tilde x^0$ as cosmological coordinate time $t$. For now we will adopt this classical cosmology notation.

To be a solution of the Finsler gravity equations we found that $f(\eta) = c_1 e^{\eta c_2}$, which implies from the coordinate change between $\eta$ and $t$ that
\begin{align}
    dt = c_1 e^{\eta c_2} d\eta\ \Leftrightarrow\ \eta(t) = \frac{1}{c_2} \ln\left(\frac{c_2}{c_1}(t-c_3)\right)\,,
\end{align}
where $c_3$ is a constant of integration. Thus, in cosmological time, the scale factor of the vacuum Finsler cosmology we find is
\begin{align}
    f(t) = c_2 (t - c_3)\,.
\end{align}

Interestingly, it turns out that these solutions are not only Ricci-flat and conformally flat (by their explicit form), but flat, in the sence that all components of the nonlinear curvature tensor $R^a{}_{bc} = \delta_b N^a_c - \delta_c N^a_a$ vanish. Nevertheless the spacetime has nontrivial geometric features.

\section{Discussion}
The solutions that we have presented above are, to the best of our knowledge, the first non-Berwaldian exact solutions to Pfeifer and Wohlfarth's field equation. Known exact solutions are scarce since in particular the Landsberg tensor terms in the field equations are difficult to understand. Employing a unicorn ansatz, i.e. non-Berwaldian Landsberg spaces, makes our solutions particularly special. We have shown that there is a subclass of our solutions for which the light cone structure is physically viable. In fact, it is equivalent to the light cone of the flat special relativistic Minkowski metric. Moreover, we have shown that one of the solutions has cosmological symmetry, i.e. it is spatially homogeneous and isotropic. Additionally, it is conformally flat, with the conformal factor depending only on the timelike coordinate, and we have shown that this conformal factor can be interpreted as the scale factor, which then turns out to be a linear function of cosmological time, leading to the natural interpretation of a linearly expanding (or contracting) Finslerian universe.
\\

As an additional curiosity that we have found that the requirement of a physically light cone structure does not strictly speaking necessitate Lorentzian signature, as is widely assumed. This is illustrated by one of our solutions, which has positive definite signature, and yet has a light cone that is equivalent to the light cone of flat Minkowski space. It is interesting and suprising that such things are apparently possible in Finsler geometry, and this paper shows the first explicit example of a (positive definite) Finsler metric with this property, which seems to be closely related to lack of smoothness of the Finsler metric in certain nontrivial subsets of $TM$. 

On the other hand, when the interior of the light cone is interpreted as the set of timelike directions, the future and past pointing unit-normalized timelike directions each form a deformed sphere with one point removed (the zero vector) rather than a deformed hyperboloid. We expect that this will affect relativistic time dilation between different observers in an interesting way and this should be further investigated in the future.\\ 

The results obtained in this paper motivate us to begin a systematic search for cosmological Landsberg spacetimes that solve the field equations, using recent results characterizing cosmological symmetry in Finsler spacetimes \cite{Hohmann:2020mgs} and Elgendi's machinery for constructing unicorns using conformal transformations \cite{elgendi_2020, Elgendi2021a}. Since (properly Finslerian) cosmological solutions of Berwald type are necessarily static \cite{Hohmann:2020mgs} any interesting such Landsberg spacetime must necessarily be a unicorn.

The exciting next step in the study of Finsler gravity, is to study unicorn solutions of the field equation sourced by the 1-particle distribution function of a kinetic gas, in homogeneous and isotropic symmetry. This scenario describes a realistic universe, filled with a kinetic gas with a nontrivial velocity distribution. As we already obtained a nontrivial solution for vacuum Finsler unicorn cosmology, more realistic, matter sourced solutions will help us to further investigate the conjecture that an accelerated expansion of the Universe is caused by the contribution of the velocity distribution of the cosmological gas, which sources a Finslerian spacetime geometry.

\begin{acknowledgments}
C.P. was funded by the cluster of excellence Quantum Frontiers funded by the Deutsche Forschungsgemeinschaft (DFG, German Research Foundation) under Germany's Excellence Strategy - EXC-2123 QuantumFrontiers - 390837967. The authors would like to acknowledge networking support by the COST Action CA18108.
\end{acknowledgments}

\appendix
\section{Proof of Prop. \ref{prop:signature}}\label{app:prf4}

\begin{proof}
We may choose coordinates such that $\phi = \phi(\hat y) = \varepsilon_1(y^1)^2+\varepsilon_2(y^2)^2+\varepsilon_3(y^3)^2$. Since the spacetime dimension is fixed to 4, wherever $F$ is sufficiently differentiable the calculation of the determinant of the fundamental tensor is in principle a straightforward exercise. It is given by
\begin{align}
    \det g 
    &=  \textrm{sgn}(\phi(\hat y))\varepsilon_1\varepsilon_2\varepsilon_3 f(x^0)^8\exp\left(\tfrac{8|y^0|}{ |y^0| + \text{sgn}(\phi)\sqrt{|\phi|} }\right)\,.
\end{align}
The determinant already gives us a pretty good idea of what the possible signature of $g_{\mu\nu}$ can be. In particular, since $g_{\mu\nu}$ is a four-dimensional matrix, it has Lorentzian signature, either of type $(+,-,-,-)$ or $(-,+,+,+)$, if and only if its determinant is negative. 
\begin{enumerate}[a)]
    \item Suppose that $\phi_{ij}$ is positive definite; then all $\epsilon_i$ and $\textrm{sgn}(\phi(\hat y))$ are positive, and hence $\det g$ is positive.
    \item Suppose that  $\phi_{ij}$ is negative definite; then all $\epsilon_i$ and $\textrm{sgn}(\phi(\hat y))$ are negative, and hence $\det g$ is positive. 
    \item Suppose that  $\phi_{ij}$ is of Lorentzian signature $(+,+,-)$; then $\det g$ is negative whenever $\textrm{sgn}(\phi(\hat y))>0$, i.e. on $\mathcal{S(\phi)}$, and $\det g$ is positive whenever $\textrm{sgn}(\phi(\hat y))<0$, i.e. on $\mathcal{T(\phi)}$.
    \item Suppose that $\phi_{ij}$ is of Lorentzian signature $(-,-,+)$; then $\det g$ is negative whenever $\textrm{sgn}(\phi(\hat y))<0$, i.e. on $\mathcal{S(\phi)}$, and $\det g$ is positive whenever $\textrm{sgn}(\phi(\hat y))>0$, i.e. on $\mathcal{T(\phi)}$.
\end{enumerate}

This already shows that $g_{\mu\nu}$ is Lorentzian if and only if $\phi$ is Lorentzian and $y\in \R\times \mathcal S(\phi)$. But the sign of the determinant does not suffice to determine whether this signature is mostly plus or mostly minus. Similarly it does not tell us much about the signature of  $g_{\mu\nu}$ when $\phi$ is positive or negative definite. In order to found, we will distinguish the following cases.\\

\subsection*{Case 1: $\phi$ Lorentzian and $y\in \R\times\mathcal S(\phi)$}
We first consider the case that $\phi$ is Lorentzian. Without loss of generality (WLOG) we set $\phi(\hat y) = \epsilon (y^1)^2+ \epsilon (y^2)^2 - \epsilon (y^3)^2$, where $\epsilon=\pm 1$. The choice of the sign $\epsilon$ selects if we are in case c) or d) from above.\\

Now note that given a vector $y\in T_xM$ that is $\phi$-spacelike, it follows from the symmetries of the Finsler metric and in particular from the $3$-dimensional Lorentz symmetry of $\phi$ that for we may always change coordinates, without changing the form of $\phi$ (and $F$), such that $y^2=y^3=0$.\\

For any choice of epsilon, by direct calculation we find, using that $\epsilon^2=1$ and $|\epsilon|=1$, that $g_{\mu\nu}$ is of the form
\begin{align}
g_{\mu\nu} = e^{\frac{2 |y^0|}{|y0|+\varepsilon|y1|}} f(x^0)^2
\begin{pmatrix}
M & 0 & 0\\
0 & 1 & 0 \\
0 & 0 & -1
\end{pmatrix}\,,
\end{align}
where $M$ is an ($\varepsilon$-dependent) positive definite $2\times 2$ matrix\footnote{$M$ must be either positive definite or negative definite, since we have already shown that $g_{ij}$ has Lorentzian signature. The fact that $M_{ab}v^av^b = ((y^0)^2+(y^1)^2)/(|y^0| + \varepsilon|y^1|)^2>0$ for $v^a = (0,1)$ thus shows that $M$ is positive definite.}. Hence we conclude that $g_{\mu\nu}$ is of the mostly plus type $(+,+,+,-)$.\\

\subsection*{Case 2: $\phi$ Lorentzian and $y\in \R\times\mathcal T(\phi)$}

In this case we may WLOG choose coordinates such that $\phi(\hat y) = -\epsilon (y^1)^2+ \epsilon (y^2)^2 + \epsilon (y^3)^2$, where $\epsilon=\pm 1$, and 
such that $y^2=y^3=0$. Again by direct calculation we find that $g_{\mu\nu}$ is of the form

\begin{align}
g_{\mu\nu} = e^{\frac{2 |y^0|}{|y0|+\varepsilon|y1|}} f(x^0)^2
\begin{pmatrix}
M & 0 & 0\\
0 & -1 & 0 \\
0 & 0 & -1
\end{pmatrix}\,,
\end{align}
where $M$ is a positive definite $2\times 2$ matrix. Hence we conclude that $g_{\mu\nu}$ is in this case of signature $(+,-,-,+)$.\\

\subsection*{Case 3: $\phi$ positive or negative definite}
In this case we may WLOG choose coordinates such that $\phi(\hat y) = \epsilon (y^1)^2+ \epsilon (y^2)^2 + \epsilon (y^3)^2$, where $\epsilon=\pm 1$, and 
such that, for any given $y\in T_xM$, we have $y^2=y^3=0$. Again by direct calculation we find that $g_{\mu\nu}$ is of the form

\begin{align}
g_{\mu\nu} = e^{\frac{2 |y^0|}{|y0|+\varepsilon|y1|}} f(x^0)^2
\begin{pmatrix}
M & 0 & 0\\
0 & 1 & 0 \\
0 & 0 & 1
\end{pmatrix}\,,
\end{align}
where $M$ is a positive definite $2\times 2$ matrix. Hence we conclude that $g_{\mu\nu}$ is positive definite.
\end{proof}


\bibliographystyle{ieeetr}
\bibliography{GeneralBib}

\end{document}